\documentclass[10pt,a4paper]{article}
\usepackage[T1]{fontenc}
\usepackage{amsmath}
\usepackage[cm]{fullpage}
\usepackage{amssymb}
\usepackage{indentfirst}
\usepackage{siunitx}
\usepackage{makecell}
\usepackage{float}

\usepackage{amsfonts}
\usepackage{graphicx}
\usepackage{amsthm}
\usepackage{enumitem}
\usepackage{mathtools}
\usepackage{bbm}
\newtheorem{theorem}{Theorem}[section]

\newtheorem{lemma}[theorem]{Lemma}

\usepackage[backend=biber,style=numeric]{biblatex}
\title{Log-linear Model for Dual System Estimation and Computational Considerations}
\author{Zhiyuan Lu}
\begin{document}
	\maketitle
	\indent Estimation of population sizes using multiple data sources has long been the practice of the Census Bureau. It is realistic to expect that even a census would not perfectly count every member of the population, and hence the use of a dual system estimator utilizing two independent population counts would provide more accurate figures. However, the use of multiple data sources cause issues when determining how much they agree. In particular, when dividing the population into distinct categories, the two sources can disagree on which category a particular person belongs to, or have missing category information on one or both sources for a recorded person. These issues have been described in the papers \cite{van2022multiple} and \cite{heijden2018overview}. This existing work has provided the framework of a statistical model and recommendation of an EM algorithm, which allows the estimation of category sizes among two data sources. 
	\newline
	\newline
	\indent A downside of using the EM algorithm is the computational time scaling rapidly as the number of categories increase: the EM algorithm can obtain quick estimates in under 1 second if there are two data sources dividing the population into two categories each, but if there are two data sources each divided into 50 categories, then the computation can require hours. An alternative method that can obtain estimates much faster is thus greatly appreciated, as it will make the analysis of population counts divided into many groups much more feasible.
	\newline
	\newline
	\indent This document will outline such a method. First the original statistical model and the EM algorithm will be described in detail, introducing notation that will be helpful in later sections. After this are sections describing the main mathematical results of this document, and conditions required for this result to hold. The most stringent condition required for the main result is that within the recorded data:
	\begin{itemize}
		\item for contingency counts recorded on both data sources, all combinations of known categories have above 0 counts
		\item for contingency counts recorded by one data source, all known categories of the recorded data source have above 0 counts.
	\end{itemize}
	Under this and several other conditions, the primary new result written here will enable a new method, one that can be argued to obtain the same answer as using the EM algorithm while requiring much less computational time. This is also confirmed in simulations, where the EM algorithm and this new algorithm obtain very close estimates when applied to randomly generated data, while the newer algorithm runs in about 1/100 of the time. A proof of the main result will finally be presented at the end of the document.

	\section{Loglinear Model of Contingency Table for Two Data Sources}
	The primary examples of contingency tables used in \cite{van2022multiple} and \cite{heijden2018overview} involved two data sources, $A$ and $B$, with data source $A$ dividing its counts among the variable $a$ and data source $B$ dividing its counts by $b$. The survey data are summarized in the form of counts for $x_{i,j,k,\ell}$, where $(i,j,k,\ell)$ are assorted values of $(A,B,a,b)$. Here, $i$ and $j$ can each take values of 0 or 1, corresponding to whether the count consists of data from register $A$ or $B$, or missing from $A$ or $B$. The category variables $a$ and $b$ can take a wider range of values, due to the potential of large number of categories and also due to counts from missing data entries. Here, the convention will be that $a$ takes values from $1,\dots,n_A$ for counts without missing information from $A$, and $-1,-2,\dots$ for counts with missing information from $A$. Similarly, $b$ will take values from $1,\dots,n_B$ for counts with no missing information from $B$ and $-1,-2,\dots$ for counts of entries with missing information from $B$.
	\begin{table}[h!]
		\centering
		\begin{tabular}{c | c | c c c | c}
			\quad	&	\quad	&	\makecell{$B=1$\\(on MOH)} 	&	\quad	&	\quad	&	\makecell{$B=0$\\(not on MOH)}  \\
			\hline
			\quad	&	\quad	&	\makecell{$b=1$\\(non-Maori)}	&	\makecell{$b=2$\\(Maori)}	&	\makecell{$b=-1$\\(ethnicity unknown)}	&	\makecell{$b=-1$\\(ethnicity unknown)}\\
			\hline
			\makecell{$A=1$ \\(on Census)}	&	\makecell{$a=1$\\(non-Maori)}	&\num{3004335}	& \num{31995}	& \num{150840}	& \num{38634}	\\
					&	\makecell{$a=2$\\(Maori)}	& \num{108189}	& \num{435465}	& \num{12405}	& \num{4368}	\\
					&	\makecell{$a=-1$\\(ethnicity unknown)}	& \num{16512}	& \num{2769}	& \num{900} 	& \num{438}		\\
			\hline
			\makecell{$A=0$\\(not on Census)} & \makecell{$a=-1$\\(ethnicity unknown)}	& \num{398838}	& \num{146976}	&	\num{24636}	&	-
		\end{tabular}
		\caption{ Example of a contingency table of counts from two sources of data, a census and ministry of health (MOH) dataset taken from New Zealand, used in \cite{van2022multiple}. Here $n_A=n_B=2$, there are 15 values of $x_{i,j,k,\ell}$ for $(i,j,k,\ell)\in Ind_{data}$. }
		\label{tab:vdh_data_examp}
	\end{table}
~\newline
	\indent The data set of contingency table counts is $\big(x_{i,j,k,\ell}: (i,j,k,\ell)\in Ind_{data}\big)$, for a collection of indices $Ind_{data}$. The aim is to take the available counts  $\big(x_{i,j,k,\ell}: (i,j,k,\ell)\in Ind_{data}\big)$ to estimate for the true counts $\big( y_{i,j,k,\ell}:(i,j,k,\ell)\in Ind_{full} \big)$, where $Ind_{full}:=\{0,1\}^2\times \{1,\dots ,n_A\}\times\{1,\dots,n_B\}$ are the indices of the complete contingency table the EM algorithm aims to estimate.
		\begin{table}[H]
		\centering
		\begin{tabular}{l|l|c c|c c}
			\quad	&	\quad	&	\makecell{$B=1$\\(on MOH)} 	&	\quad	&		\makecell{$B=0$\\(not on MOH)} \quad	& \\
			\hline
			\quad	&	\quad	&	\makecell{$b=1$\\(non-Maori)}	&	\makecell{$b=2$\\(Maori)}	&	\makecell{$b=1$\\(non-Maori)}	&	\makecell{$b=2$\\(Maori)}\\
			\hline
			\makecell{$A=1$ \\(on Census)}	&	\makecell{$a=1$\\(non-Maori)}	& \num{3170294.8}	& \num{33787.9}	& \num{38616.0}	& \num{411.6}	\\
			\quad 	&	\makecell{$a=2$\\(Maori)}	& \num{111242.5}	& \num{448084.8}	& \num{877.6}	& \num{3534.9}	\\
			\hline
			\makecell{$A=0$\\(not on Census)} & \makecell{$a=1$\\(non-Maori)}	& \num{402709.4}	& \num{10770.8}	& \num{4905.2}	& \num{131.2} \\
			\quad &	\makecell{$a=2$\\(Maori)}	& \num{14130.7}	& \num{142839.1}	& \num{111.5}	&	\num{1126.8}
		\end{tabular}
		\caption{Estimated values of $y_{i,j,k,\ell}$ for the 16 values of $(i,j,k,\ell)\in Ind_{full}$, using the raw data of $x_{i,j,k,\ell}$'s presented in Table \ref{tab:vdh_data_examp}. These values were obtained from \cite{van2022multiple}.  }
		\label{tab:vdh_output}
	\end{table}
	~\newline
	~\indent Usually, the index sets $Ind_{full}$ and $Ind_{data}$ are different sets with different indices as elements. In order to turn counts indexed by the former into counts indexed to the latter, a function is defined in order to link the collection of coordinates:
	\begin{equation}
		CD(i,j,k,\ell):=\left\{ \text{ indices in }Ind_{full}\text{ where the counts }x_{i,j,k,\ell}\text{ can be distributed to } \right\}.
	\end{equation}
	In practical settings, if a coordinate $\mathbf{u}\in Ind_{data}$ corresponds to a category without missing information, $CD(\mathbf{u})$ would just map to a singleton set (usually consisting of $\mathbf{u}$ itself; if $\mathbf{u}$ corresponds to a category with missing information, then $CD(\mathbf{u})$ maps to a set of all coordinates in $Ind_{full}$ that counts in $\mathbf{u}$ potentially can be distributed to. For example, in a data set where both registers $A$ and $B$ are split by 2 income levels (1-poor, 2-rich), then (1,1,1,-1) corresponds to someone recorded as rich in $A$ but of unknown income in $B$. People counted in this category could be of either category in $B$, and hence a possible assignment is $CD(\mathbf{u})=\{ (1,1,1,1),(1,1,1,2) \}$.
	\begin{table}[h!]
		\centering
		\begin{tabular}{c|c||c|c}
			\makecell{$(i,j,k,\ell)$\\value}		&	$CD(i,j,k,\ell)$	&	\makecell{$(i,j,k,\ell)$\\value}		&	$CD(i,j,k,\ell)$	\\
			\hline
			(1,1,1,1)		&	$\{(1,1,1,1)\}$										&	(1,1,1,2)		&	$\{(1,1,1,2)\}$		\\
			(1,1,2,1)		&	$\{(1,1,2,1)\}$										&	(1,1,2,2)		&		$\{(1,1,2,2)\}$		\\
			(1,1,1,-1)		&	$\{ (1,1,1,1),(1,1,1,2) \}$			&	(1,1,-1,1)	&	$\{ (1,1,1,1),(1,1,2,1) \}$	\\
			(1,1,2,-1)		&	$\{ (1,1,2,1),(1,1,2,1) \}$	&			(1,1,-1,2)		&	$\{ (1,1,1,2),(1,1,2,2) \}$				\\
			(1,1,-1,-1)		&	$\{ (1,1,1,1),(1,1,1,2),(1,1,2,1),(1,1,2,2) \}$		& 	&	\\
			(1,0,1,-1)		&	$\{ (1,0,1,1),(1,0,1,2) \}$			&	(1,0,2,-1)		&	$\{ (1,0,2,1),(1,0,2,2) \}$	\\
			(1,0,-1,-1)		&	$\{(1,0,1,1),(1,0,1,2),(1,0,2,1),(1,0,2,2)\}$	&	&	\\
			(0,1,-1,1)		&	$\{ (0,1,1,1),(0,1,2,1) \}$			&	(0,1,-1,2)		&	$\{ (0,1,1,2),(0,1,2,2) \}$	\\
			(0,1,-1,-1)		&	$\{(0,1,1,1),(0,1,1,2),(0,1,2,1),(0,1,2,2)\}$	&	&
		\end{tabular}
		\caption{ Elements of $Ind_{data}$ and which set they are mapped to using the $CD$ function, in the context of the data and indices shown in Tables \ref{tab:vdh_data_examp} and \ref{tab:vdh_output}.}
		\label{tab:CD_func_examples}
	\end{table}
	~\newline
	\indent A loglinear model is assumed for the set of true counts $\big( y_{i,j,k,\ell}:(i,j,k,\ell)\in Ind_{full} \big)$: each $y_{i,j,k,\ell}$ is generated from a Poisson distribution, with the mean determined by a set of variables corresponding to the \textbf{maximal} model described in \cite{heijden2018overview} and \cite{van2022multiple}. This relation is expressed as
	\begin{equation}\label{eq:count_lambda_mean}
		\log \left(\mathbb{E}\left[ y_{i,j,k,\ell} \right]\right)\sim \lambda_0+\lambda^A_i+\lambda^B_j+\lambda^{a}_k+\lambda^b_\ell+\lambda^{Ab}_{i,\ell}+\lambda^{Ba}_{j,k}+\lambda^{ab}_{k\ell},
	\end{equation}
	wtih the additional restriction that only the following variables are allowed to be nonzero:
	\begin{eqnarray}\label{eq:free_lambdas}
		&&\lambda_0,\lambda^A_1,\lambda^B_1,\nonumber\\
		&&\lambda^a_k\text{ for }k=2,\dots,n_A,\nonumber\\
		&&\lambda^b_\ell\text{ for }\ell=2,\dots,n_B,\nonumber\\
		&&\lambda^{Ab}_{1,\ell}\text{ for }\ell=2,\dots,n_B,\nonumber\\
		&&\lambda^{Ba}_{1,k}\text{ for }k=2,\dots,n_A,\nonumber\\
		&&\lambda^{ab}_{kl}\text{ for }k=2,\dots,n_A,\quad \ell=2,\dots,n_B,
	\end{eqnarray}
	for a total of $n_An_B+n_A+n_B$ free variables. Another equivalent way of expressing the relation in (\ref{eq:count_lambda_mean}), which will be used for the rest of this document, would be to consider the vector of all the variables in (\ref{eq:free_lambdas}):
	\begin{eqnarray}
		\mathbf{\Lambda}:=\left( \lambda_0,\lambda^A_1,\dots,\lambda^{ab}_{n_A,n_B}  \right),
	\end{eqnarray}
	where dimension of $\mathbf{\Lambda}$ is $\dim(\mathbf{\Lambda})=n_An_B+n_A+n_B$, and a function $VM:Ind_{full}\to\{ 0,1 \}^{\dim(\mathbf{\Lambda})}$ (short for Variable Match), so that (\ref{eq:count_lambda_mean}) can be rewritten as:
	\begin{eqnarray}\label{eq:counts_log_means}
		\log \left(\mathbb{E}\left[ y_{i,j,k,\ell} \right]\right)=VM(i,j,k,\ell)\cdot \mathbf{\Lambda}\qquad \text{ for }(i,j,k,\ell)\in Ind_{full},
	\end{eqnarray}
	which is in the form of a Poisson regression problem.

	\section{EM Algorithm for Estimation}
	In general, the EM algorithm functions as an iterative algorithm on statistical models involving a set of data with missing values, another set of latent data representing the truth, and a parameterized distribution function. In this situation, $\mathbf{X}:=\big(x_{\mathbf{v}}:\mathbf{v}=(i,j,k,\ell)\in Ind_{data} \big)$ is the incomplete observed data and $\mathbf{Y}:=\big(y_{\mathbf{v}}:\mathbf{v}=(i,j,k,\ell)\in Ind_{full}\big)$ is a complete set of counts. The distribution function of the latter set is characterized by independent Poisson distributions with log means $VM(\mathbf{v})\cdot \mathbf{\Lambda}$ (refer back to (\ref{eq:counts_log_means})). The full distribution of $\mathbf{Y}$ is:
	\begin{eqnarray}
		p(\mathbf{Y}|\mathbf{\Lambda})=\prod_{\mathbf{v}\in Ind_{full} } \frac{\exp\left[ y_{\mathbf{v}}VM(\mathbf{v})\cdot \mathbf{\Lambda}-\exp\left( VM(\mathbf{v})\cdot \mathbf{\Lambda} \right) \right]  }{y_{\mathbf{v}}!}.
	\end{eqnarray}
	which leads to the log-likelihood function:
	\begin{eqnarray}\label{eq:log-likelihood_full}
		\log p(\mathbf{Y}|\mathbf{\Lambda})=\sum_{\mathbf{v}\in Ind_{full} }\Big[y_{\mathbf{v}}VM(\mathbf{v})\cdot \mathbf{\Lambda}-\exp\left( VM(\mathbf{v})\cdot \mathbf{\Lambda} \right)\Big]+C(\mathbf{Y})
	\end{eqnarray}
	for some function $C()$. For the observed counts $\mathbf{X}$, details of their distribution functions are not explicitly stated in \cite{heijden2018overview} and \cite{van2022multiple}, beyond the assumption on ``missingness at random". 
	\newline
	\newline
	\indent For any $\mathbf{u}\in Ind_{data}$, each individual that contributes to the count $x_\mathbf{u}$ can have their true affiliation belong to any member of $CD(\mathbf{u})$; if all of the information about this entry was observed, then it would be seen as a part of the count $y_\mathbf{v}$ for some $\mathbf{v}\in CD(\mathbf{u})$. This makes $x_\mathbf{u}$ a sum of ``contributions'' from the values $\{y_\mathbf{v}:\mathbf{v}\in CD(\mathbf{u})\}$. With respect to how this relates to the calculation of the EM algorithm, the numerical results of \cite{heijden2018overview} and \cite{van2022multiple} appears to act under the assumption that conditional on the value of $x_{\mathbf{v}}$:
	\begin{eqnarray}\label{eq:expected_contribution}
		\mathbb{E}\left[ \text{contribution from }y_\mathbf{v}\text{ in }x_\mathbf{u}\big| \mathbf{X},\mathbf{Y} \right]=\frac{x_{\mathbf{u}}y_\mathbf{v}}{\sum\limits_{\mathbf{w}\in CD(\mathbf{u})}y_\mathbf{w}}.
	\end{eqnarray}
	The reasoning appears to be that with $y_\mathbf{v}$ as independent Poisson variables, this makes $x_\mathbf{u}$ Poisson and the contributions multinomial when conditioned on $x_\mathbf{u}$.
	\newline
	\newline
	\indent Using these foundations, the EM algorithm to estimate for the true values of $\mathbf{\Lambda}$ can be fully detailed. Suppose one starts with initial guess $\mathbf{\hat{\Lambda}}^{(0)}$ for the parameter $\mathbf{\Lambda}$, then for $t\in\mathbb{N}$ the algorithm proceeds by taking the previous estimates $\mathbf{\hat{\Lambda}}^{(t-1)}$ and repeatedly applying the following 2 steps:
	\begin{itemize}
		\item[E-step: ] First, calculate the temporary values
		\begin{eqnarray}
			\hat{y}_{\mathbf{v}}^{(t)*}&=&\mathbb{E}\left[ y_{\mathbf{v}} \Big| \mathbf{\Lambda}=\mathbf{\hat{\Lambda}}^{(t-1)} \right]\nonumber\\
			&=&\exp\left[ VM(\mathbf{v})\cdot \mathbf{\hat{\Lambda}}^{(t-1)} \right]\\
			&&\text{ for }\mathbf{v}\in Ind_{full}.
		\end{eqnarray}
		Next, proceed to assign the values of $\mathbb{E}\left[ y_{\mathbf{v}} \big| \mathbf{X},\mathbf{\hat{\Lambda}}^{(t-1)} \right]$'s to the $\hat{y}_{\mathbf{v}}^{(t)}$'s through the following steps:
		\begin{enumerate}
			\item For every $\mathbf{v}\in Ind_{full}$ such that there exists a $\mathbf{u}\in Ind_{data}$ where $\mathbf{v}\in CD(\mathbf{u})$, set $\hat{y}^{(t)}_{\mathbf{v}}\leftarrow 0$.
			\item For all other $\mathbf{v}\in Ind_{full}$ where $\mathbf{v}\notin \bigcup\limits_{\mathbf{u}\in Ind_{data}} CD(\mathbf{u})  $, set $\hat{y}^{(t)}_{\mathbf{v}}\leftarrow \hat{y}^{(t)*}_{\mathbf{v}}.$
			\item For every element $\mathbf{u}\in Ind_{data}$, set 
			\begin{equation}\label{eq:missing_update}
				\hat{y}_\mathbf{v}^{(t)}\leftarrow \hat{y}_\mathbf{v}^{(t)}+\frac{x_{\mathbf{u}}\hat{y}_\mathbf{v}^{(t)*}}{\sum\limits_{\mathbf{w}\in CD(\mathbf{u})}\hat{y}_\mathbf{w}^{(t)*}}
			\end{equation}
			for all $\mathbf{v}\in CD(\mathbf{u})$
		\end{enumerate}
		with expression (\ref{eq:missing_update}) motivated by the expression in (\ref{eq:expected_contribution}). Using (\ref{eq:log-likelihood_full}), the conditional expectation of the log-likelihood is:
		\begin{eqnarray}\label{eq:q_param_function}
			Q\left( \mathbf{\Lambda} \big| \mathbf{\hat{\Lambda}}^{(t-1)} \right)&=&\mathbb{E}_{\mathbf{Y}\sim p(\cdot| \mathbf{X},\mathbf{\hat{\Lambda}}^{(t-1)} ) }\Big[ \log p\left(\mathbf{X},\mathbf{Y}\Big| \mathbf{\Lambda}\right) \Big]\nonumber\\
			&=&\sum_{\mathbf{v}\in Ind_{full} }\Big[\hat{y}_{\mathbf{v}}^{(t)}VM(\mathbf{v})\cdot \mathbf{\Lambda}-\exp\left( VM(\mathbf{v})\cdot \mathbf{\Lambda} \right)\Big]+C^*(\mathbf{\hat{\Lambda}}^{(t-1)})
		\end{eqnarray}
		for some function $C^*()$.
		
		\item[M-step: ] Maximize the expression in (\ref{eq:q_param_function}) with respect to $\mathbf{\Lambda}$:
		\begin{eqnarray}
			\mathbf{\hat{\Lambda}}^{(t)}&=&\underset{\mathbf{\Lambda}}{\arg\max}\,\, Q\left( \mathbf{\Lambda} \big| \mathbf{\hat{\Lambda}}^{(t-1)} \right)\nonumber\\
			&=&\underset{\mathbf{\Lambda}}{\arg\max}\sum_{\mathbf{v}\in Ind_{full} }\Big[\hat{y}_{\mathbf{v}}^{(t)}VM(\mathbf{v})\cdot \mathbf{\Lambda}-\exp\left( VM(\mathbf{v})\cdot \mathbf{\Lambda} \right)\Big].
		\end{eqnarray} 
		This is identical to a Poisson regression problem involving the covariates $\big(VM(\mathbf{v})\big)_{\mathbf{v}\in Ind_{full}}$ and response variable $\big( \hat{y}_\mathbf{v}^{(t)}  \big)_{\mathbf{v}\in Ind_{full}}$.
	\end{itemize}
	~\newline
	\indent The EM algorithm calculates the MLE for the parameter of the likelihood function, which is $\mathbf{\Lambda}$ for this setting. However, users of the log-linear model on population data may be more interested in estimates of the contingency table counts, the $y_\mathbf{v}$'s. For this task it can be more useful to think of the EM algorithm operating in three steps: suppose one has a set of estimates $\{\hat{y}_\mathbf{v}^{(t-1)}  \}_{\mathbf{v}\in Ind_{full}}$, then the EM algorithm steps for procuring the next iteration of count estimates are:
	\begin{enumerate}
		\item Obtain an estimate $\mathbf{\hat{\Lambda}}^{(t-1)}$ from $\big( \hat{y}_\mathbf{v}^{(t-1)} \big)_{\mathbf{v}\in Ind_{full}} $ through:
		\begin{eqnarray}
			\mathbf{\hat{\Lambda}}^{(t-1)}
			&=&\underset{\mathbf{\Lambda}}{\arg\max}\sum_{\mathbf{v}\in Ind_{full} }\Big[\hat{y}_{\mathbf{v}}^{(t-1)}VM(\mathbf{v})\cdot \mathbf{\Lambda}-\exp\left( VM(\mathbf{v})\cdot \mathbf{\Lambda} \right)\Big].
		\end{eqnarray}
		This step can be described as applying a function $PR:\mathbb{R}^{|Ind_{full}|}\to \mathbb{R}^{\dim( \mathbf{\Lambda} )}$ (short for Poisson Regression) such that 
		\begin{equation}
			\mathbf{\hat{\Lambda}}^{(t-1)}=PR\left( \big( \hat{y}_\mathbf{v}^{(t-1)} \big)_{\mathbf{v}\in Ind_{full}} \right).
		\end{equation}
		\item Obtain a set of temporary counts $\big(\hat{y}^{(t)*}_\mathbf{v}\big)_{\mathbf{v}\in Ind_{full}}$ from $\mathbf{\hat{\Lambda}}^{(t-1)}$ through calculating:
		\begin{eqnarray}
			\hat{y}_{\mathbf{v}}^{(t)*}=\exp\left[ VM(\mathbf{v})\cdot \mathbf{\hat{\Lambda}}^{(t-1)} \right]\qquad\text{ for }\mathbf{v}\in Ind_{full}.
		\end{eqnarray}
		This step can be described as applying a function $UE:\mathbb{R}^{\dim( \mathbf{\Lambda} )}\to\mathbb{R}^{|Ind_{full}|}$ (short for Unconditioned Expectation) such that 
		\begin{equation}
			\big(\hat{y}_{\mathbf{v}}^{(t)*}\big)_{\mathbf{v}\in Ind_{full}}=UE\left( \mathbf{\hat{\Lambda}}^{(t-1)} \right).
		\end{equation}
		\item From the temporary values $\hat{y}^{(t)*}_\mathbf{v}$, obtain the next iteration of counts $\{\hat{y}_\mathbf{v}^{(t)}  \}_{\mathbf{v}\in Ind_{full}}$ through the steps:
		\begin{enumerate}
			\item Initialize $\hat{y}_{\mathbf{v}}^{(t)}=0$ for all $\mathbf{v}\in Ind_{full}$.
			\item For all $\mathbf{v}\in Ind_{full}$,
			\begin{eqnarray}\label{eqn:DIST_rewritten}
				\hat{y}_\mathbf{v}^{(t)}=\begin{cases} \hat{y}^{(t)*}_{\mathbf{v}} \qquad\qquad\qquad\qquad  \text{ if }\mathbf{v}\notin CD(\mathbf{u})  \text{ for every }\mathbf{u}\in Ind_{data}\\
					\sum
					\limits_{\mathbf{u}\in CD^-(\mathbf{v}) } \left[ \left(x_\mathbf{u} \hat{y}_{\mathbf{v}}^{(t)*} \right)  \left( \sum\limits_{\mathbf{w}\in CD(\mathbf{u})}\hat{y}_\mathbf{w}^{(t)*} \right)^{-1} \right]\qquad \text{otherwise}
				\end{cases}
			\end{eqnarray}
			where 
			\begin{equation}
				CD^-( \mathbf{v} ):=\left\{ \mathbf{u}\in Ind_{data}:\, \mathbf{v}\in CD(\mathbf{u}  )\,\,\text{and}\,\, x_{\mathbf{u}}>0 \right\}.
			\end{equation}
		\end{enumerate}
		The operation mapping $\big( \hat{y}_\mathbf{v}^{(t)*} \big)_{\mathbf{v}\in Ind_{full}}$ to $\big( \hat{y}_\mathbf{v}^{(t)} \big)_{\mathbf{v}\in Ind_{full}}$ can also be described as a function $DIST:\mathbb{R}^{|Ind_{full}|}\to \mathbb{R}^{|Ind_{full}|}$ (short for count DISTribution).
	\end{enumerate} 
	Written this way, the EM algorithm can be envisioned as a composite function $DIST\circ UE \circ PR$, which iterates the count estimates by:
	\begin{equation}
		\big(\hat{y}_\mathbf{v}^{(t)}\big)_{\mathbf{v}\in Ind_{full}}=  DIST\left(UE\left(PR\left(\big(\hat{y}_\mathbf{v}^{(t-1)}\big)_{\mathbf{v}\in Ind_{full}}  \right)\right)\right)
	\end{equation}


	~\newline
	\indent Repeated iterations of the EM algorithm in order to find a point of convergence would be the same steps as finding a fixed point for $DIST\circ UE \circ PR$ by fixed point iteration. The existence of an attracting fixed point\footnote{a fixed point $y_{fix}$ that has a neighbor $U$, where any point within $U$ used as a starting point for the fixed point iteration algorithm would cause convergence to $y_{fix}$} for the composite function would be synonymous with the EM algorithm being able to converge if selecting the right initial guess. The foxus from here on would be concentrated on fixed points, for if there exists a way to find a fixed point of $DIST\circ UE \circ PR $, then that point is a very good candidate for the converging point of the EM algorithm.
	\newline
	\newline
	\indent A computational issue presents itself when performing the EM algorithm, at the $PR$ step of the iterative process. At this step a Poisson regression is performed to estimate for $n_An_B+n_A+n_B$ parameters, and this optimization procedure is typically performed using gradient descent. Conventional methods of gradient descent, such as Newton's method, works by inverting the Hessian matrix, which in this case would be a square matrix of width $n_An_B+n_A+n_B$. This means that every step of gradient descent would require $O(n_A^3n_B^3)$ operations, meaning the total computational time of each $PR$ step is $O(S\cdot n_A^3n_B^3)$, where $S$ is the number of require gradient descent steps that can vary depending on the specific algorithm used and the current location of the $\hat{y}^{(t)}_\mathbf{v}$'s that the $PR$ function is applied to. This scaling, which is at least the order of the amount of data cubed, will cause the overwhelming majority of computational time in practical scenarios, and can cause the entire algorithm to require hours of runtime when $n_A$ and $n_B$ have been increased to as low as 50. Comparatively, in most data examples encountered by the author, the $UE$ and $DIST$ steps require only $O(n_An_B)$ computational steps at every EM iteration.
	\newline
	\newline
	\indent Of course, this analysis does not take into account of optimization procedures that do not depend on Hessian inversion, such as stochastic gradient descent (SGD) or adaptive moment estimation (ADAM). The order of computation might be different with these other procedures, but the problem remains that each $PR$ step involves optimization on $O(n_An_B)$ variables, and $PR$ would be repeated at each iteration of the EM algorithm. The previous analysis provides a motivation to cut the time cost specifically for the $PR$ step, for this would massively reduce the running time of the EM algorithm as a whole.
	
	\section{Alternative Fixed Point Method}\label{sec:alt_fix_point_method}
	Refer back to the New Zealand data given in Table \ref{tab:vdh_data_examp} that would give the EM estimates in Table \ref{tab:vdh_output}, there are some clear patterns that are noticeable from the final results in the latter table. These patterns can be more easily discerned by splitting the table into four quadrants, and the indices $Ind_{full}$ into four distinct subsets. For ease of notation, define the four sets
	\begin{equation}
		Ind_{i,j}:=\Big\{ (i,j,k,\ell):\,k\in \{1,\dots,n_A\},\ell\in \{1,\dots,n_B\} \Big\}\qquad \text{ for }(i,j)\in \{0,1\}^2.
	\end{equation}
	In the quadrant of $Ind_{1,0}$, it is clear that values alongside the rows are in the same proportion as in the quadrant of $Ind_{1,1}$, as $3,170,294.8/33,787.9\approx 38,616.0/411.6$ and $111,242.5/448,084.8\approx 877.6/3534.9$. Similarly, $3,170,294.8/111,242.5\approx 402,709.4/14,130.7$ and $33,787.9/448,084.8\approx 10,770.8/142,839.1$, shows the $Ind_{0,1}$ quadrant has values in the same proportion as $Ind_{1,1}$ quadrant, alongside its columns. Finally, the $Ind_{0,0}$ quadrant can be directly calculated as the product of values inside the $Ind_{0,1}$ and $Ind_{1,0}$ quadrants, divided by the values of the $Ind_{1,1}$ quadrant; e.g. $4905.2\approx 402709.4\cdot 38616.0/3170294.8$.
	\newline
	\newline
	\indent These patterns are not surprising, because the Poisson model in (\ref{eq:count_lambda_mean}) do imply these relations. This suggests that if there is way to determine the EM algorithm converging point values in the $Ind_{1,1}$ quadrant, the row totals of the $Ind_{1,0}$ quadrant, and the column totals of the $Ind_{0,1}$ quadrant, then it is possible to determine all $4n_An_B$ values for every coordinate in $Ind_{full}$ from these $n_An_B+n_A+n_B$ values.
	\newline
	\newline
	\indent There does seem to be a pattern that is helpful in determining the values for the $Ind_{1,1}$ quadrant. Observe that if there are counts $\mathbf{z}=\left(z_{\mathbf{v}}\right)_{\mathbf{v}\in Ind_{full}}$, then the application of the $DIST$ function upon those counts give the following algebraic relation:
	\begin{eqnarray}\label{eq:vdh_example_quandrant_expressions}
		\pi_{1,1,1,1}\left( DIST(\mathbf{z}) \right)&=& x_{1,1,1,1}+\frac{x_{1,1,1,-1}z_{1,1,1,1}}{z_{1,1,1,:}}+\frac{x_{1,1,-1,1}z_{1,1,1,1}}{z_{1,1,:,1}}+\frac{x_{1,1,-1,-1}z_{1,1,1,1}}{\sum\limits_{\mathbf{v}\in Ind_{1,1}}z_\mathbf{v} }\nonumber\\
		\pi_{1,1,1,2}\left( DIST(\mathbf{z}) \right)&=& x_{1,1,1,2}+\frac{x_{1,1,1,-1}z_{1,1,1,2}}{z_{1,1,1,:}}+\frac{x_{1,1,-1,2}z_{1,1,1,2}}{z_{1,1,:,2}}+\frac{x_{1,1,-1,-1}z_{1,1,1,2}}{\sum\limits_{\mathbf{v}\in Ind_{1,1}}z_\mathbf{v} }\nonumber\\
		\pi_{1,1,2,1}\left( DIST(\mathbf{z}) \right)&=& x_{1,1,2,1}+\frac{x_{1,1,2,-1}z_{1,1,2,1}}{z_{1,1,2,:}}+\frac{x_{1,1,-1,1}z_{1,1,2,1}}{z_{1,1,:,1}}+\frac{x_{1,1,-1,-1}z_{1,1,2,1}}{\sum\limits_{\mathbf{v}\in Ind_{1,1}}z_\mathbf{v} }\nonumber\\
		\pi_{1,1,2,2}\left( DIST(\mathbf{z}) \right)&=& x_{1,1,2,2}+\frac{x_{1,1,2,-1}z_{1,1,2,2}}{z_{1,1,2,:}}+\frac{x_{1,1,-1,2}z_{1,1,2,2}}{z_{1,1,:,2}}+\frac{x_{1,1,-1,-1}z_{1,1,2,2}}{\sum\limits_{\mathbf{v}\in Ind_{1,1}}z_\mathbf{v} }
	\end{eqnarray}
	where $\pi_{i,j,k,\ell}$ function extracts the value at coordinate $(i,j,k,\ell)$ (i.e., $\pi_\mathbf{v}(\mathbf{z}))=z_\mathbf{v}$ for every $\mathbf{v}\in Ind_{full}$), $z_{i,j,:,\ell}:=\sum\limits_{c=1}^{n_A}z_{i,j,c,\ell}$ and $z_{i,j,k,:}:=\sum\limits_{c=1}^{n_B}z_{i,j,k,c}$ are row and column totals, and the $x_{1,1,k,\ell}$ values are the upper right nine entries in the upper left of Table \ref{tab:vdh_data_examp}. Imagining a function that maps $( z_{1,1,1,1}, z_{1,1,1,2},z_{1,1,2,1},z_{1,1,2,2})$ to the four expressions in (\ref{eq:vdh_example_quandrant_expressions}) as $DIST_{match}:\mathbb{R}^4\to\mathbb{R}^4$, it can be observed that the $Ind_{1,1}$ entries of Table \ref{tab:vdh_output} are nearly the fixed points of $DIST_{match}$:
	\begin{eqnarray}
		\big(3170294.8,\,33787.9,\,111242.5,\,448084.8\big)\approx DIST_{match}\big(3170294.8,\,33787.9,\,111242.5,\,448084.8\big).
	\end{eqnarray}
	In this case, it seems that the most important coordinates of the MLE are also the fixed points of the simpler arithmetic operations in (\ref{eq:vdh_example_quandrant_expressions}). This suggests an alternative way to determine the $Ind_{1,1}$ values of the MLE: find a solution to $\mathbf{z}=DIST_{match}(\mathbf{z})$. This might be difficult analytically, as (\ref{eq:vdh_example_quandrant_expressions}) shows this would involve solving a system of four polynomials, each of order 4, but it may be very tractable computationally using fixed point iteration, i.e., starting with an initial guess point and applying the $DIST_{match}$ function repeatedly until convergence. This alternative procedure would run as follows:
	\begin{enumerate}
		\item From an initial guess of counts for coordinates in $Ind_{1,1}$, apply a function equaling the effect of $DIST$ on those coordinates until convergence (such a function was labeled $DIST_{match}$ in the previous paragraph).
		\item Apply a similar iterative process to find column totals of quadrant $Ind_{0,1}$ and row totals of quadrant $Ind_{1,0}$ that are fixed by $DIST$.
		\item Fill out the entries in $Ind_{0,1}$ by making them proportionate to the columns of quadrant $Ind_{1,1}$, and the entries in $Ind_{1,0}$ proportionate to the rows of $Ind_{1,1}$ and making sure the row/column totals equal to what was found in the previous step.
		\item Fill out the entries in $Ind_{0,0}$ as a product of the $Ind_{1,0}$ and $Ind_{0,1}$ quadrants divided by the $Ind_{1,1}$ quadrant.
	\end{enumerate}
	The question becomes whether this procedure can generally give the same estimate of the MLE as the EM algorithm. Under what conditions can this alternative method work?

	\section{Assumptions}
	A general statement can in fact be made about the aforementioned algorithm: under certain conditions, if the algorithm has a fixed point then that point is also a fixed point of the EM algorithm. Before stating this result, the necessary conditions and associated operations in the algorithm need to be described.
	\newline
	\newline
	\indent The result will need the following list of assumptions, which will be called Structural Assumptions since they pertain to the structure of $Ind_{data}$ and how it relates to the $Ind_{full}$ set through the $CD$ function.
	\begin{enumerate}[label=S\arabic*]
		\item \label{asn:maxmodel} The underlying loglinear model is maximal, with the log means following expression (\ref{eq:count_lambda_mean}) utilizing all of the variables in (\ref{eq:free_lambdas}).
		\item \label{asn:distdivide} For any element in $\mathbf{v}=(i,j,k,\ell)\in Ind_{data}$, any element $\mathbf{u}\in CD(i,j,k,\ell)$ must be of the form $\mathbf{u}=(i,j,k',\ell')$ for some $1\leq k'\leq n_A,\,1\leq\ell'\leq n_B$. In other words the first two elements of $\mathbf{v}$ must be equal to the first two entries of every element in $CD(\mathbf{v})$ for every $\mathbf{v}\in Ind_{data}$.
		\item \label{asn:no_missing_totals} No coordinates of $Ind_{0,0}$ belongs in any subset mapped to by the $ CD $ function; for every $1\leq k\leq n_A$ and $1\leq \ell\leq n_B$, $(0,0,k,\ell)\notin CD(\mathbf{v})$ for every $\mathbf{v}\in Ind_{data}$.
		\item \label{asn:nonmatchdist} Suppose $(0,1,k,\ell)\in Ind_{data}$ and $(0,1,k',\ell')\in CD(0,1,k,l)$ for some $1\leq k'\leq n_A,\,1\leq\ell'\leq n_B$, then $(0,1,c,\ell')\in CD(0,1,k,\ell)$ for every $1\leq c \leq n_A$. Similarly, if $(1,0,k,\ell)\in Ind_{data}$ and $(1,0,k',\ell')\in CD(1,0,k,l)$, then $(1,0,k',c)\in CD(1,0,k,l)$ for all $1\leq c\leq n_B$.
	\end{enumerate}
	~\indent The models presented in \cite{van2022multiple} allowed flexibility in using different loglinear models in order to fit the data, so it is worth stating again in Assumption \ref{asn:maxmodel} that the following results only holds for the maximal model. The other assumptions have stronger real world foundations.
	\newline
	\newline
	\indent Assumption \ref{asn:distdivide} is practically always true, since entries counted within register $A$ (or $B$) cannot contribute to counts of entries outside of register $A$ (or $B$). For instance, counts of the form $x_{1,1,k,\ell}$ consist of entries found from both $A$ and $B$, and all of these entries must be distributed towards counts of the form $y_{1,1,k',\ell'}$. For any of these entries to be distributed to a $y_{1,0,k',\ell'}$ would mean those entries are actually not in $B$, which is nonsensical. Similar logic holds for counts of the form $x_{1,0,k,\ell}$ and $x_{0,1,k,\ell}$. Assumption \ref{asn:no_missing_totals} can be reasoned in a similar way, and also as a natural consequence of Assumption \ref{asn:distdivide}. Entries missing from both registers $A$ and $B$ will not be recorded in the data, and there will be no counts of the form $x_{0,0,k,\ell}$, and hence no values will be distributed to any estimates of the form $y_{0,0,k',\ell'}$.
	\newline
	\newline
	\indent Assumption \ref{asn:nonmatchdist} also has strong real world grounding. Entries that contribute to a count $x_{1,0,k,\ell}$ are those missing from registry $B$, and therefore there are no information on which category $1,\dots,n_B$ that they belong to. If such entries have a possibility of residing within $y_{1,0,k',\ell'}$ for some $k'$ and $\ell'$, but not be possible for them to reside within $y_{1,0,k',\ell''}$ for some $\ell''\neq \ell'$, then it would be contradictory as that would mean there is information on which categories within $B$ these entries are members of. 
	\newline
	\newline
	\indent Working under Assumptions \ref{asn:nonmatchdist} and \ref{asn:distdivide}, the $CD$ function has a particular effect on the the row/column totals of the $Ind_{1,0}$ and $Ind_{0,1}$ quadrants. To denote this, define the functions $CD_{B=0}:Ind_{1,0}\to 2^{\{1,\dots,n_A\}}$ and $CD_{A=0}:Ind_{0,1}\to 2^{\{1,\dots,n_B\}}$ as
	\begin{eqnarray}\label{eq:CD_marginal_def}
		CD_{B=0}(\mathbf{v})&:=&\bigg\{ k\in \{1,\dots,n_A\} : (1,0,k,1)\in CD(\mathbf{v}) \bigg\}\nonumber\\
		CD_{A=0}(\mathbf{v})&:=&\bigg\{ \ell \in \{1,\dots,n_B\} : (0,1,1,\ell)\in CD(\mathbf{v}) \bigg\},
	\end{eqnarray}
	and associated pseudo-inverse functions $CD_{B=0}^-:\{1,\dots,n_A\}\to 2^{Ind_{1,0}}$ and $CD_{A=0}^-:\{1,\dots,n_B\}\to 2^{Ind_{0,1}}$ as
	\begin{eqnarray}
		CD_{B=0}^-(k)&:=& \big\{  \mathbf{u}\in Ind_{1,0}: k\in CD_{B=0}(\mathbf{u})\,\,\text{and}\,\, x_{\mathbf{u}}>0 \big\}\nonumber\\
		CD_{A=0}^-(\ell)&:=& \big\{  \mathbf{u}\in Ind_{0,1}: \ell\in CD_{A=0}(\mathbf{u})\,\,\text{and}\,\, x_{\mathbf{u}}>0 \big\}
	\end{eqnarray}
	for all $k\in \{1,\dots,n_A\}$ and $\ell\in \{1,\dots,n_B\}$. 
	\newline
	\newline
	\indent Now suppose there is a set of non-negative real values $\mathbf{z}=\big( z_{\mathbf{v}} \big)_{\mathbf{v}\in Ind_{full}}$, where the application of the $DIST$ function yields $\mathbf{z}'$ that equals
	\begin{equation}
		\mathbf{z}'=\left( z_\mathbf{v}' \right)_{\mathbf{v}\in Ind_{full}}=DIST\left( \mathbf{z}  \right),
	\end{equation}
	then the marginal totals of the $\mathbf{z}'$ satisfy
	\begin{eqnarray}\label{eq:marginal_totals}
		z_{1,0,k,:}' &=& \begin{cases}z_{1,0,k,:}\qquad\qquad\qquad\qquad\qquad\qquad\text{ if } k\notin CD_{B=0}(u)\text{ for all }\mathbf{u}\in Ind_{data}\\
			\sum\limits_{\mathbf{u}\in CD_{B=0}^-(k) } \left[  \left(x_{\mathbf{u}} z_{1,0,k,:}\right) \left(    \sum\limits_{c\in CD_{B=0}(\mathbf{u}) }z_{1,0,c,:}   \right)^{-1} \right]\qquad\text{ otherwise}
		\end{cases}\nonumber\\
		z_{0,1,:,\ell}' &=& \begin{cases}z_{0,1,:,\ell}\qquad\qquad\qquad\qquad\qquad\qquad\text{ if } \ell\notin CD_{A=0}(\mathbf{u})\text{ for all }\mathbf{u}\in Ind_{data}\\
			\sum\limits_{\mathbf{u}\in CD_{A=0}^-(\ell) } \left[  \left(x_{\mathbf{u}} z_{0,1,:,\ell}\right) \left(    \sum\limits_{c\in CD_{A=0}(\mathbf{u}) }z_{0,1,:,c}   \right)^{-1} \right]\qquad\text{ otherwise},
		\end{cases}
	\end{eqnarray}
	for $1\leq k\leq n_A$ and $1\leq \ell \leq n_B$. For details of why this is true, see Lemma \ref{lem:DIST_marginals}. These expressions are very reminiscent of the $DIST$ function as expressed in (\ref{eqn:DIST_rewritten}), and therefore motivates the definition of functions $DIST_{B=0}$ and $DIST_{A=0}$, the former mapping $ \mathbb{R}^{n_A} $ to $\mathbb{R}^{n_A}$ and the latter mapping $ \mathbb{R}^{n_B} $ to itself, where 
	\begin{eqnarray}
		DIST_{B=0}\left( z_{1,0,1,:},\dots, z_{1,0,n_A,:} \right)&=&\left( z_{1,0,1,:}',\dots,z_{1,0,n_A,:}' \right)\nonumber\\
		DIST_{A=0}\left(  z_{0,1,:,1},\dots, z_{0,1,:,n_B}  \right)&=&\left( z_{0,1,:,1}',\dots,z_{0,1,:,n_B}' \right),
	\end{eqnarray}
	as the $z'_\mathbf{v}$s are defined in (\ref{eq:marginal_totals}).
	~\newline
	~\newline 
	\indent Another implication of Assumption \ref{asn:distdivide} is that it effectively maps indices within $Ind_{1,1}$ to itself. To be specific, for entries in $Ind_{1,1}=\{ (1,1,k,\ell):(k,\ell)\in \{1,\dots,n_A\}\times \{1,\dots n_B\} \}$, which are of the form $(1,1,k,\ell)$, the only entries $\mathbf{u}$ within $Ind_{data}$ for which $(1,1,k,\ell)\in CD(\mathbf{u})$, must have the first two entries of $\mathbf{u}$ be $(1,1)$, and hence every element in $CD(\mathbf{u})$ must also have the first two entries equaling $(1,1)$. Therefore,
	\begin{eqnarray}
		z_{1,1,k,\ell}'&=&\sum_{\substack{\mathbf{u}:\, (1,1,k,\ell)\in CD(\mathbf{u})\\ \mathbf{u}\in Ind_{data}} } \left[ \left(x_\mathbf{u} z_{1,1,k,\ell} \right)  \left( \sum_{\mathbf{w}\in CD(\mathbf{u})}z_\mathbf{w} \right)^{-1} \right]\nonumber\\
		&&\sum_{\substack{\mathbf{u}:\, (1,1,k,\ell)\in CD(\mathbf{u})\\ \mathbf{u}\in Ind_{data}} } \left[ \left(x_\mathbf{u} z_{1,1,k,\ell} \right)  \left( \sum_{(1,1,c,b)\in CD(\mathbf{u})}z_{1,1,c,b} \right)^{-1}\right],
	\end{eqnarray}
	in other words, every $z_{1,1,k,\ell}'$ for $(1,1,k,\ell) \in Ind_{1,1}$ depend only on $\left(z_\mathbf{w}\right)_{\mathbf{w}\in Ind_{1,1} }$, and it is possible to define a function $DIST_{match}: Ind_{1,1} \to Ind_{1,1}$ where the output of $DIST_{match}$ is identical to the entries of $DIST$ on the coordinates corresponding to $Ind_{1,1}$:
	\begin{eqnarray}
		DIST_{match}\left( \left(z_\mathbf{v}\right)_{\mathbf{v}\in Ind_{1,1}} \right)=\left( DIST\left( \left( z_\mathbf{v} \right)_{\mathbf{v}\in Ind_{full}} \right) \right)_{Ind_{1,1}}
	\end{eqnarray}

	\begin{theorem}\label{thm:main_result}
		Suppose Assumptions $\ref{asn:maxmodel}$, $\ref{asn:distdivide}$, \ref{asn:no_missing_totals}, and $\ref{asn:nonmatchdist}$, are true. Next suppose there are 
		\begin{itemize}
			\item positive values $\left(y^{(s)}_\mathbf{v} \right)_{\mathbf{v}\in Ind_{1,1}}$ that are fixed points of $DIST_{match}$:
			\begin{equation}
				DIST_{match}\left( \left(y^{(s)}_\mathbf{v}  \right)_{\mathbf{v}\in Ind_{1,1}} \right) =\left(y^{(s)}_\mathbf{v}\right)_{\mathbf{v}\in Ind_{1,1}}
			\end{equation}
			
			\item positive values $\left( y^{(s)}_{1,0,1,:},\dots,y^{(s)}_{1,0,n_A,:} \right)$ that are fixed points of $DIST_{B=0}$:
			\begin{equation}
				DIST_{B=0}\left( y^{(s)}_{1,0,1,:},\dots,y^{(s)}_{1,0,n_A,:} \right)=\left( y^{(s)}_{1,0,1,:},\dots,y^{(s)}_{1,0,n_A,:} \right)
			\end{equation}
			\item positive values $\left( y^{(s)}_{0,1,:,1},\dots,y^{(s)}_{0,1,:,n_B} \right)$ that are fixed points of $DIST_{A-0}$.
		\end{itemize}
		Then it is possible to construct a set of values $\mathbf{y}^{(s)}=\left(y^{(s)}_\mathbf{v}\right)_{\mathbf{v}\in Ind_{full}}$ which is a fixed point of both $DIST$ and $UE\circ PR$, and hence is also a fixed point of the EM algorithm:
		\begin{itemize}
			\item for $\mathbf{v}\in Ind_{1,1}$, let $\pi_\mathbf{v}(\mathbf{y}^{(s)})= y^{(s)}_\mathbf{v}$
			\item for entries indexed by coordinates within $Ind_{1,0}$, let 
			\begin{equation}\label{eq:unmatch_B0_fixed_coords}
				y_{1,0,k,\ell}^{(s)}:=\pi_{1,0,k,\ell}\left( \mathbf{y}^{(s)} \right)=y_{1,0,k,:}^{(s)}\left(\frac{ y_{1,1,k,\ell}^{(s)} }{ \sum_{c=1}^{n_B} y_{1,1,k,c}^{(s)} }\right)
			\end{equation}
			for $(k,\ell)\in \{1,\dots,n_A\}\times\{1,\dots,n_B\}$
			\item for entries indexed by coordinates within $Ind_{0,1}$, let 
			\begin{eqnarray}\label{eq:unmatch_A0_fixed_coords}
				y_{0,1,k,\ell}^{(s)}:=\pi_{0,1,k,\ell}\left( \mathbf{y}^{(s)} \right)=y_{0,1,:,\ell}^{(s)}\left(\frac{ y_{1,1,k,\ell}^{(s)} }{ \sum_{c=1}^{n_A} y_{1,1,c,\ell}^{(s)} }\right)
			\end{eqnarray}
			for $(k,\ell)\in \{1,\dots,n_A\}\times\{1,\dots,n_B\}$
			\item after calculating the above, let 
			\begin{eqnarray}\label{eq:unmatch_fixed_coords}
				y_{0,0,k,\ell}^{(s)}:=\frac{ y_{1,0,k,\ell}^{(s)}\cdot y_{0,1,k,\ell}^{(s)} }{y_{1,1,k,\ell}^{(s)}}
			\end{eqnarray}
			for $(k,\ell)\in \{1,\dots,n_A\}\times\{1,\dots,n_B\}$
		\end{itemize}
	\end{theorem}
	\begin{proof}
		See section \ref{sec:main_proof}.
	\end{proof}
	This result gives justification for the proposed algorithm given at the end of Section \ref{sec:alt_fix_point_method} under a range of conditions. The requirements for this result to hold are the assumptions \ref{asn:maxmodel} to \ref{asn:nonmatchdist}, which have real world grounding and should generally hold with actual data sets, and the more restrictive condition of the existence of fixed points of $DIST_{match}$, $DIST_{A=0}$, and $DIST_{B=0}$ that have only positive values as entries. This more restrictive assumption becomes a major issue if not held, for if all fixed points of the three functions contain a zero entry, then the fixed points constructed by following the instructions of the previous result will not be a fixed point of $UE\circ PR$, since $UE$ can only output all positive entries. In this scenario, the constructed point will still be a fixed point of $DIST$, but may not be a fixed point of the EM algorithm $DIST\circ UE\circ PR$.
	\newline
	\newline
	\indent A first response to this problem will be to find out under which conditions do the fixed points of $DIST_{match}$, $DIST_{A=0}$, and $DIST_{B=0}$ contain only positive entries. The exact conditions necessary are not currently known, but the following Positive Assumptions are sufficient to guarantee that the fixed points of $DIST_{match}$, $DIST_{A=0}$, and $DIST_{B=0}$ have only positive entries:
	\begin{enumerate}[label=P\arabic*]
		\item \label{asn:p_match} For every $\mathbf{v}\in Ind_{1,1}$, there exists a $\mathbf{u}\in Ind_{data}$ such that $CD(\mathbf{u})=\{ \mathbf{v} \}$ and $x_\mathbf{u}>0$.
		\item \label{asn:p_B0}For every $1\leq k\leq n_A$, there exists a $\mathbf{u}\in Ind_{data}$ where $CD_{B=0}(\mathbf{u})=\{ (1,0,k,1),\dots,(1,0,k,n_B) \}$ and $x_{\mathbf{u}}>0$.
		\item \label{asn:p_A0}For every $1\leq \ell \leq n_B$, there exists a $\mathbf{u}\in Ind_{data}$ where $CD_{A=0}(\mathbf{u})=\{ (0,1,1,\ell),\dots,(1,0,n_A,\ell) \}$ and $x_{\mathbf{u}}>0$.
	\end{enumerate}
	A more grounded example of these conditions can be seen in the real world data example presented in Table \ref{tab:CD_func_examples}. There, any $(1,1,k,\ell)\in Ind_{1,1}$ has $CD(1,1,k,\ell)=\{(1,1,k,\ell)\}$ for \ref{asn:p_match}, any $1\leq k\leq n_A$ has the coordinate $\mathbf{u}=(1,0,k,-1)$ for \ref{asn:p_B0}, and any $1\leq \ell\leq n_B$ has coordinate $\mathbf{u}=(0,1,-1,\ell)$ for \ref{asn:p_A0}. As seen in Table \ref{tab:vdh_data_examp}, all of these $x_\mathbf{u}$ values are positive, and hence the fixed points, seen in Table \ref{tab:vdh_output}m contain all positive values.
	\newline
	\newline
	\indent Unfortunately, \ref{asn:p_match} to \ref{asn:p_A0} cannot be argued to always hold in a practical setting in the same way that was done for conditions \ref{asn:distdivide} to \ref{asn:nonmatchdist}. An example scenario in which the former conditions do not hold is if the total true number of people counted in $A$ and/or $B$ are small and the number of categories $n_A$ and/or $n_B$ are large. If a thousand people are counted in data source $A$, and is divided among 5 ethnicities, 5 age groups, 4 economic categories, and 2 sexes, then the counts in $A$ would be spread among $n_A=5\times 5\times 4\times 2=200$ categories, and there might be categories, say ethnicity 2, extremely young category 1, extremely rich 4, sex 1, that have a total of 0 counts among $A$, and hence $x_{1,1,k,\ell}=0$ for all $\ell$ and $k$ corresponding to this specific set of categories. In summary, dividing the counts in $A$ or $B$ into too many or too unbalanced categories can cause some counts to be 0, and no longer guarantee that this alternative algorithm for the log-linear MLE can work mathematically.
	\newline
	\newline
	\indent Practically, it is not truly necessary for the algorithm to get the exact same result as the EM algorithm, it is sufficient if it yields an answer that is close enough. It may be the case that for fixed sets $Ind_{data}$ and $Ind_{full}$, and fixed function $CD$ to link the two sets, then there is a continuous function that maps $(x_\mathbf{u})_{ \mathbf{u}\in Ind_{data} }$ to a fixed point of the EM algorithm, and another continuous function that maps $(x_\mathbf{u})_{ \mathbf{u}\in Ind_{data} }$ to the point described in Theorem \ref{thm:main_result}. If this is true, then either of the following can yield a point that is very close to the fixed point of the EM algorithm:
	\begin{itemize}
		\item Use the construction proposed in Theorem \ref{thm:main_result} even if it contains entries of 0.
		\item Use the construction proposed in Theorem \ref{thm:main_result} on the data $(x_\mathbf{u}')_{ \mathbf{u}\in Ind_{data} }$, where $x_{\mathbf{u}}'=\max\{ \delta,x_{\mathbf{u}} \}$ for some small positive value $\delta$.
	\end{itemize}
	In order to assess the accuracy of such methods, some numerical studies will need to be done.

		\section{Proof of Theorem \ref{thm:main_result}}\label{sec:main_proof}
		\begin{proof}
			The first part of the proof will show that $\mathbf{y}^{(s)}$ is a fixed point of $DIST$, and the second part will show it is a fixed point of $UE\circ PR$.
			\newline
			\newline
			\textbf{Part 1}  Verifying that $\mathbf{y}^{(s)}$ is a fixed point of the $DIST$ step of the EM algorithm can be done by checking on the different types of components. For every $\mathbf{v}\in Ind_{full}$, the associated value at that coordinate, after the application of $DIST$ on the value of $y_\mathbf{v}^{(s)}$ is 
		\begin{equation}
			\pi_\mathbf{v}\left( DIST\left(\mathbf{y}^{(s)}\right) \right),
		\end{equation}
		which in turn equals:
		\begin{itemize}
			\item if $\mathbf{v}\in Ind_{1,1}$, then since $(y_{\mathbf{w}}^{(s)})_{\mathbf{w}\in Ind_{1,1}}$ is a fixed point of $DIST_{match}$:
			\begin{eqnarray}
				\pi_\mathbf{v}\left( DIST\left(\mathbf{y}^{(s)}\right) \right)&=&\pi_\mathbf{v}\left( DIST_{match}\left( \left( y^{(s)}_\mathbf{u} \right)_{\mathbf{u}\in Ind_{1,1}} \right) \right)\nonumber\\
				&=&\pi_\mathbf{v}\left( \left( y^{(s)}_\mathbf{u} \right)_{\mathbf{u}\in Ind_{1,1}} \right)\nonumber\\
				&=& y^{(s)}_\mathbf{v}
			\end{eqnarray}
			\item if $\mathbf{v}\in Ind_{0,0}$, then by Assumption \ref{asn:no_missing_totals}, there is no $\mathbf{u}\in Ind_{data}$ for which $\mathbf{v}\in CD(\mathbf{u})$, and hence $DIST$ fixes this coordinate by defintion:
			\begin{equation}
				\pi_\mathbf{v}\left( DIST\left(\mathbf{y}^{(s)}\right) \right)=y^{(s)}_\mathbf{v}
			\end{equation}
			\item if $\mathbf{v}\in Ind_{1,0}\cup Ind_{0,1}$, then since $\mathbf{y}^{(s)}$ has only strictly positive entries and satisfies \ref{asn:operational}, and the marginal totals $y^{(s)}_{1,0,k,:}$'s are fixed by $DIST_{B=0}$ and $y^{(s)}_{0,1,:,\ell}$'s are fixed by $DIST_{A=0}$, Lemma \ref{lem:DIST_marginal_fix} means that every element indexed by $Ind_{1,0}$ and $Ind_{0,1}$ are fixed by $DIST$:
			\begin{equation}
				\pi_\mathbf{v}\left( DIST\left(\mathbf{y}^{(s)}\right) \right)=y^{(s)}_\mathbf{v}.
			\end{equation}
		\end{itemize}
		Overall, this shows that
		\begin{eqnarray}
			DIST\left( \mathbf{y}^{(s)} \right)=\mathbf{y}^{(s)}
		\end{eqnarray}
		\newline
		\newline
		\textbf{Part 2} When performing the Poisson regression in the $PR$ step of the EM algorithm, the goal would be to find the parameters
		\begin{eqnarray} 
			\boldsymbol{\Lambda}&=&\Big[ \lambda_0,\lambda_1^A,\lambda_1^B, \lambda^a_{2},\dots,\lambda^a_{n_A},\lambda^b_2,\dots,\lambda^b_{n_B},\nonumber\\
			&&\lambda^{Ab}_{1,2},\dots,\lambda^{Ab}_{1,n_B},\lambda^{Ba}_{1,2},\dots,\lambda^{Ba}_{1,n_A},
			\lambda^{ab}_{2,2},\dots,\lambda^{ab}_{n_A,n_B} \Big]^T
		\end{eqnarray}
		for the Poisson regression where $\mathbf{y}^{(s)}$ is the response variable and the covariate variable can be constructed using the coordinates of $Ind_{full}$. In this case, it is possible to find a deterministic solution to the Poisson regression, because there is a unique solution to the following set of equations:
		\begin{eqnarray}\label{eq:full_logmean_equations}
			\log y_{1,1,1,1}^{(s)}&=& \lambda_0+\lambda_1^A+\lambda_1^B\nonumber\\
			\log y_{1,1,k,1}^{(s)} &=& \lambda_0+\lambda_1^A+\lambda_1^B+\lambda^a_k+\lambda^{Ba}_{1,k}\nonumber\\
			\log y_{1,1,1,\ell}^{(s)} &=& \lambda_0+\lambda_1^A+\lambda_1^B+\lambda^b_\ell+\lambda_{1,\ell}^{Ab}  \nonumber\\
			\log y_{1,1,k,\ell}^{(s)} &=& \lambda_0+\lambda_1^A+\lambda_1^B+\lambda^a_k+\lambda^b_\ell+\lambda^{Ab}_{1,\ell}+\lambda^{Ba}_{1,k}+\lambda^{ab}_{k,\ell}\nonumber\\
			\log y_{1,0,1,1}^{(s)} &=& \lambda_0+\lambda_1^A\nonumber\\
			\log y_{1,0,k,1}^{(s)} &=& \lambda_0+\lambda_1^A + \lambda^a_k \nonumber\\
			\log y_{1,0,1,\ell}^{(s)} &=& \lambda_0+\lambda_1^A +\lambda_\ell^b+\lambda^{Ab}_{1,\ell} \nonumber\\
			\log y_{1,0,k,\ell}^{(s)} &=& \lambda_0+\lambda_1^A +\lambda^a_k+\lambda^b_\ell+\lambda^{Ab}_{1,\ell}+\lambda^{a,b}_{k,\ell} \nonumber\\
			\log y_{0,1,1,1}^{(s)} &=& \lambda_0+\lambda^B_1\nonumber\\
			\log y_{0,1,k,1}^{(s)} &=& \lambda_0+\lambda^B_1+\lambda^a_k+\lambda^{Ba}_{1,k}\nonumber\\
			\log y_{0,1,1,\ell}^{(s)} &=& \lambda_0+\lambda^B_1+\lambda^b_\ell\nonumber\\
			\log y_{0,1,k,\ell}^{(s)} &=& \lambda_0+\lambda^B_1+\lambda^a_k+\lambda^{Ba}_{1,k}+\lambda^b_\ell+\lambda^{ab}_{k,\ell}\nonumber\\
			\log y_{0,0,1,1}^{(s)} &=& \lambda_0\nonumber\\
			\log y_{0,0,k,1}^{(s)} &=& \lambda_0+\lambda^a_k \nonumber\\
			\log y_{0,0,1,\ell}^{(s)} &=& \lambda_0+\lambda^b_\ell \nonumber\\
			\log y_{0,0,k,\ell}^{(s)} &=& \lambda_0+\lambda^a_k+\lambda^b_\ell+\lambda^{ab}_{k,\ell}
		\end{eqnarray}
		for $2\leq k\leq n_A$, $2\leq \ell\leq n_B$. To see this, first focus on indices of $Ind_{1,1}$, and consider the substitution
		\begin{eqnarray}
			v_0&:=& \lambda_0+\lambda_1^A+\lambda_1^B\nonumber\\
			v^a_{k}&:=& \lambda^a_k+\lambda_{1,k}^{Ba}\nonumber\\
			v^b_\ell &:=& \lambda^b_\ell+\lambda_{1,\ell}^{Ab}\nonumber\\
			&\text{for }&2\leq k\leq n_A,\, 2\leq \ell\leq n_B,
		\end{eqnarray}
		which leads to the equations in the first four lines of (\ref{eq:full_logmean_equations}) to be rewritten as
		\begin{eqnarray}\label{eq:match_matrix_eq}
			\begin{bmatrix}
				1 & 0\,\dots & \dots & \dots\,0  \\[1em]
				\mathbf{1}_{n_A-1} &  \mathbb{I}_{n_A-1} &  \mathbf{0} &  \mathbf{0}   \\[1em]
				\mathbf{1}_{n_B-1} & \mathbf{0} & \mathbb{I}_{n_B-1}  &  \mathbf{0} \\[1em]
				\mathbf{1}_{(n_A-1)(n_B-1) } & \mathbf{M}_1 & \mathbf{M}_2 &  \mathbb{I}_{(n_A-1)(n_B-1)}
			\end{bmatrix}
			\begin{bmatrix}
				v_0\\[0.5em]
				v^a_2\\[0.5em]
				\vdots\\[0.5em]
				v^a_{n_A}\\[0.5em]
				v^b_2\\[0.5em]
				\vdots\\[0.5em]
				v^b_{n_B}\\[0.5em]
				\lambda^{ab}_{2,2}\\[0.5em]
				\vdots\\[0.5em]
				\lambda^{ab}_{n_A,n_B}
			\end{bmatrix}
			=
			\begin{bmatrix}
				\log y_{1,1,1,1}^{(s)} \\[0.5em]
				\log y_{1,1,2,1}^{(s)} \\[0.5em]
				\vdots \\[0.5em]
				\log y_{1,1,n_A,1}^{(s)} \\[0.5em]
				\log y_{1,1,1,2}^{(s)} \\[0.5em]
				\log y_{1,1,1,n_B}^{(s)} \\[0.5em]
				\log y_{1,1,2,2}^{(s)} \\[0.5em]
				\vdots \\[0.5em]
				\log y_{1,1,n_A,n_B}^{(s)}
			\end{bmatrix}
		\end{eqnarray}
		where for any $c\in \mathbb{N}$, $\mathbf{1}_c$ is the $c\times 1$ dimensional vectors where every entry is 1, $\mathbb{I}_c$ is the identity matrix of size $c\times c $, $\mathbf{0}$'s are matrices of all 0's and of appropriate dimensions, and $\mathbf{M}_1$ and $\mathbf{M}_2$ are some matrices. The matrix on the left of (\ref{eq:match_matrix_eq}) is clearly non-singular, as it is lower-triangular and all diagonal entries are 1. Therefore, there is a unique solution 
		\begin{equation}
			v_0^*, v_2^{a*},\dots, v_{n_A}^{a*}, v_2^{b*},\dots ,v_{n_B}^{b*},\lambda_{2,2}^{ab*},\dots,\lambda_{n_A,n_B}^{ab*}.
		\end{equation}
		to the linear system of equations in (\ref{eq:match_matrix_eq}). To find a solution to the rest of the $\lambda$ variables, first note that lines 5-12 of (\ref{eq:full_logmean_equations}) are equivalent to
		\begin{eqnarray}\label{eq:logycomp1}
			\log y_{1,0,1,\ell}^{(s)} &=& \log y_{1,1,1,\ell}^{(s)}-\lambda_1^B\nonumber\\
			\log y_{1,0,k,\ell}^{(s)} &=& \log y_{1,1,k,\ell}^{(s)}-\lambda^B_1-\lambda^{Ba}_{1,k}
		\end{eqnarray}
		for all $2\leq k\leq n_A$ and $1\leq \ell \leq n_B$, and
		\begin{eqnarray}\label{eq:logycomp2}
			\log y_{0,1,k,1}^{(s)} &=&\log y_{1,1,k,1}^{(s)}-\lambda_1^A\nonumber\\
			\log y_{0,1,k,\ell}^{(s)}&=&\log y_{1,1,k,\ell}^{(s)}-\lambda_1^A-\lambda_{1,\ell}^{Ab}
		\end{eqnarray}
		for all $1\leq k\leq n_A$ and $2\leq \ell \leq n_B$. Therefore the following set of solutions
		\begin{eqnarray}
			\lambda_1^{B*} = \log \frac{y^{(s)}_{1,1,1,:}}{y^{(s)}_{1,0,1,:}}\qquad\qquad
			\lambda_{1,k}^{Ba*} = \log \frac{ y^{(s)}_{1,1,k,:} }{ y^{(s)}_{1,0,k,:}} -\lambda_1^{B*} 
		\end{eqnarray}
		for $2\leq k\leq n_A$, plugged into right side of(\ref{eq:logycomp1}) would yield
		\begin{eqnarray}
			\log y_{1,1,1,\ell}^{(s)}-\lambda_1^{B*} &=& \log \left( y_{1,1,1,\ell}^{(s)}\frac{y_{1,0,1,:}^{(2)}}{ y^{(s)}_{1,1,1,:}}  \right)=y_{1,0,1,\ell}^{(s)}\nonumber\\
			\log y^{(s)}_{1,1,k,\ell}-\lambda_1^{B*}-\lambda_{1,k}^{Ba*} &=& \log \left( y_{1,1,k,\ell}^{(s)}\frac{ y_{1,0,k,:}^{(s)} }{ y_{1,1,k,:}^{(s)} } \right)=y_{1,0,k,\ell}^{(s)}
		\end{eqnarray}
		for all such $k$ and all $1\leq \ell\leq n_B$, where the second equality on each line is due to the definition given in (\ref{eq:unmatch_B0_fixed_coords}). Similarly, 
		\begin{eqnarray}
			\lambda_1^{A*} = \log \frac{y^{(s)}_{1,1,:,1}}{y^{(s)}_{0,1,:,1}}\qquad\qquad
			\lambda_{1,\ell}^{Ab*} = \log \frac{ y^{(s)}_{1,1,:,\ell} }{ y^{(s)}_{0,1,:,\ell}} -\lambda_1^{A*} 
		\end{eqnarray}
		for $2\leq \ell\leq  n_B$, plugged into the right side of (\ref{eq:logycomp2})  would lead to 
		\begin{eqnarray}
			\log y_{1,1,k,1}^{(s)}-\lambda_1^{A*}&=& \log\left( y_{1,1,k,1}^{(s)}\frac{y_{0,1,:,1}^{(2)}}{ y^{(s)}_{1,1,:,1}}  \right)=y_{0,1,k,1}^{(s)}\nonumber\\
			\log y^{(s)}_{0,1,k,\ell}-\lambda_1^{A*}-\lambda_{1,\ell}^{Ab*} &=& \log \left( y_{1,1,k,\ell}^{(s)}\frac{ y_{0,1,:,\ell}^{(s)} }{ y_{1,1,:,\ell}^{(s)} } \right)=y_{0,1,k,\ell}^{(s)}
		\end{eqnarray}
		for all such $\ell $ and $1\leq k\leq n_A$, using the definitions of (\ref{eq:unmatch_A0_fixed_coords}) for the second equality on each line. 
		\newline
		\newline
		\indent Altogether, and with the addition of letting $\lambda_{0}^*=v^*-\lambda_1^{A*}-\lambda_1^{B*}$, this gives values
		\begin{eqnarray}
			\boldsymbol{\Lambda}^*=\left( \lambda_0^*,\lambda_1^{A*},\dots,\lambda_{n_A,n_B}^{ab*} \right)
		\end{eqnarray}
		that satisfy lines 1 to 12 of (\ref{eq:full_logmean_equations}). Lines 13 to 16 are then immediately also satisfied, since these lines of (\ref{eq:full_logmean_equations}) leads to the equalities
		\begin{eqnarray}
			\log y_{0,0,k,\ell}^{(s)}=\log y_{1,0,k,\ell}^{(s)}-\log y_{0,1,k,\ell}^{(s)} -\log y_{1,1,k,\ell}^{(s)}
		\end{eqnarray}
		for $1\leq k\leq n_A$ and $1\leq \ell\leq n_B$, which completely agree with how the $y_{0,0,k,\ell}^{(s)}$'s are defined in (\ref{eq:unmatch_fixed_coords}). 
		\newline\newline
		\indent The $PR$ step of the EM algorithm is a Poisson regression involving the the response values $y_{i,j,k,\ell}^{(s)}$'s and the covariate matrix $\mathbf{V}$, whose rows are defined as the coefficients of the $\lambda$'s in (\ref{eq:full_logmean_equations}), or equivalently, as $VM(\mathbf{v})$ for $\mathbf{v}\in Ind_{full}$ as $VM$ was described in (\ref{eq:counts_log_means}). For finding the fit $\boldsymbol{\Lambda}$, the previous steps show there is an actual equality of 
		\begin{equation}\label{eq:Lambdastar_logy}
			\mathbf{V}\boldsymbol{\Lambda}^*= \log \mathbf{y}^{(s)}
		\end{equation}
		In addition, it can be verified that the rows of $\mathbf{V}$ span all of $\mathbb{R}^{(n_An_B+n_A+n_B)}$, so it can be concluded using Lemma \ref{lem:poisson_solution} that $\boldsymbol{\Lambda}^*$ is the solution of the Poisson regression:
		\begin{equation}
			PR\left(\mathbf{y}^{(s)}\right)=\boldsymbol{\Lambda}^*
		\end{equation}
		However from (\ref{eq:Lambdastar_logy}):
		\begin{eqnarray}
			UE\left( \boldsymbol{\Lambda}^* \right)=\exp\left[\mathbf{V} \Lambda^*\right] =\mathbf{y}^{(s)}.
		\end{eqnarray}
		This shows that $UE\left(PR \left( \mathbf{y}^{(s)} \right) \right)=\mathbf{y}^{(s)}$.

	\end{proof}

	\section{Supporting Results}
	\indent In order to prevent difficulty arising out of division by 0 scenarios, a certain operating assumption must hold throughout either the EM algorithm or the new proposed algorithm:
	\begin{enumerate}[label=OP\arabic*]
		\item \label{asn:operational} A set of values $\mathbf{z}=\left(z_\mathbf{v}\right)_{\mathbf{v}\in Ind_{full}} $ satisfies this assumption if 
		\begin{itemize}
			\item $z_\mathbf{v}\geq 0$ for every $\mathbf{v}\in Ind_{full}$
			\item for every $\mathbf{u}\in Ind_{data}$ where $x_\mathbf{u}>0$, there exists a $\mathbf{v}\in CD(\mathbf{u})$ where $z_\mathbf{v}>0$
		\end{itemize}
	\end{enumerate}
	In practice this condition is almost always satisfied throughout the EM algorithm or the proposed faster algorithm, provided either algorithm do not start with an unreasonable starting point (such as letting the starting guess count be 0 for all entries of $Ind_{full}$). To see this, it can be first observed that if $\mathbf{z}=\left(z_\mathbf{v}\right)_{\mathbf{v}\in Ind_{full}} $ satisfies \ref{asn:operational}, then so does $DIST\left( \mathbf{z} \right)$: given a $\mathbf{u}\in Ind_{data}$ such that $x_\mathbf{u}>0$ and $z_{\mathbf{v}^*}>0$ for some $\mathbf{v}^*\in CD(\mathbf{u})$, then the $\mathbf{v}^*$ entry of $DIST\left( \mathbf{z} \right)$ would also be positive: 
	\begin{eqnarray}
		\pi_{\mathbf{v}*}\left( DIST(\mathbf{z}) \right)&=&\sum_{\mathbf{s}\in CD^-(\mathbf{v}^*) } \left[ \left(x_{\mathbf{s}} z_{\mathbf{v}^*} \right)  \left( \sum_{\mathbf{w}\in CD(\mathbf{s})}z_{\mathbf{w}} \right)^{-1} \right]\nonumber\\
		&\geq& x_\mathbf{u} z_{\mathbf{v}^*}\left( \sum_{\mathbf{w}\in CD(\mathbf{u})}z_{\mathbf{w}} \right)^{-1}\nonumber\\
		&&(\mathbf{u}\in CD^-(\mathbf{v}^*)\text{ by definition, since }\mathbf{v}^*\in CD(\mathbf{v})\text{ and }x_{\mathbf{u}}>0)\nonumber\\
		&>& 0
	\end{eqnarray}
	where the last line is due to the summation term on the second line being a positive value (all $z_{ \mathbf{w} }\geq0$, and $z_{\mathbf{v}^*}>0$ is part of the summation). 
	\newline
	\newline
	\indent This shows that the $DIST$ function preserves the \ref{asn:operational} property, and similar arguments can show that this also holds for $DIST_{match}$, $DIST_{A=0}$, and $DIST_{B=0}$. 
	Any algorithm that depends on repeated applications of these functions will have their output at every step keeping the \ref{asn:operational} property, as long as the initial starting point satisfies \ref{asn:operational}. As for the EM algorithm, it has two intermediate steps, $PR$ and $UE$, between applications of the $DIST$ function. Fortunately, this pose no issue, for the $UE$ function outputs entries that are the means of Poisson distributions, which can only be all positive, and hence trivially satisfy \ref{asn:operational}.
	
	\begin{lemma}\label{lem:DIST_marginals}
		Suppose Assumptions \ref{asn:distdivide} and \ref{asn:nonmatchdist} are true. Let $(z_\mathbf{v})_{\mathbf{v}\in Ind_{full}}$ be any set of values satisfying \ref{asn:operational}, upon which if the $DIST$ function is applied would yield the values
		\begin{equation}
			\left(z_\mathbf{v}'\right)_{\mathbf{v}\in Ind_{full}}=DIST\bigg( \left(z_\mathbf{v}\right)_{\mathbf{v}\in Ind_{full}} \bigg).
		\end{equation}
		The marginal totals of these resulting values would be as described in expression (\ref{eq:marginal_totals}): for any $k\in \{1,\dots,n_A\}$ and $\ell\in \{1,\dots, n_B\}$,
		\begin{eqnarray}
			z_{1,0,k,:}' &=& \begin{cases}z_{1,0,k,:}\qquad\qquad\qquad\qquad\qquad\qquad\text{ if } k\notin CD_{B=0}(u)\text{ for all }\mathbf{u}\in Ind_{data}\\
				\sum\limits_{\mathbf{u}\in CD_{B=0}^-(k) } \left[  \left(x_{\mathbf{u}} z_{1,0,k,:}\right) \left(    \sum\limits_{c\in CD_{B=0}(\mathbf{u}) }z_{1,0,c,:}   \right)^{-1} \right]\qquad\text{ otherwise}
			\end{cases}\nonumber\\
			z_{0,1,:,\ell}' &=& \begin{cases}z_{0,1,:,\ell}\qquad\qquad\qquad\qquad\qquad\qquad\text{ if } \ell\notin CD_{A=0}(\mathbf{u})\text{ for all }\mathbf{u}\in Ind_{data}\\
				\sum\limits_{\mathbf{u}\in CD_{A=0}^-(\ell) } \left[  \left(x_{\mathbf{u}} z_{0,1,:,\ell}\right) \left(    \sum\limits_{c\in CD_{A=0}(\mathbf{u}) }z_{0,1,:,c}   \right)^{-1} \right]\qquad\text{ otherwise},
			\end{cases}
		\end{eqnarray}
	\end{lemma}
	\begin{proof}
		Suppose that for some $k\in\{1,\dots ,n_A\}$, $k\notin CD_{B=0}(\mathbf{u})$ for all $\mathbf{u}\in Ind_{data}$. By definition that means $(1,0,k,1)\notin CD(\mathbf{u})$ for all $\mathbf{u}\in Ind_{data}$, and hence for each $1\leq c\leq n_B$, $(1,0,k,c)\notin CD(\mathbf{u})$ for all $\mathbf{u}\in Ind_{data}$ through assumption \ref{asn:nonmatchdist}. Therefore $DIST$ fixes the coordinates $(1,0,k,c)$ for every $1\leq c\leq n_B$ and
		\begin{equation}
			z_{1,0,k,:}'=\sum_{c=1}^{n_B} z_{1,0,k,c}'=\sum_{c=1}^{n_B} z_{1,0,k,c}=z_{1,0,k,:}
		\end{equation}
		~\newline
		For any $k\in\{1,\dots ,n_A\}$ and $c\in\{1,\dots,n_B\}$ where $k\in CD_{B=0}(\mathbf{u})$ for some $\mathbf{u}\in Ind_{data}$, use the expression (\ref{eqn:DIST_rewritten}) to write 
		\begin{equation}
			z_{1,0,k,c}'=\sum_{ \mathbf{u}\in CD^-(1,0,k,c) }  \left[ \left(x_\mathbf{u} z_{1,0,k,c} \right)  \left( \sum_{\mathbf{w}\in CD(\mathbf{u})}z_\mathbf{w} \right)^{-1} \right].
		\end{equation}
		Summing over $c$ means that for any $k\in\{1,\dots ,n_A\}$,
		\begin{eqnarray}\label{eq:zmarginal_sum}
			z_{1,0,k,:}'&=& \sum_{c=1}^{n_B} \left[ \sum_{ \mathbf{u}:\in CD^-(1,0,k,c) }  \left[ \left(x_\mathbf{u} z_{1,0,k,c} \right)  \left( \sum_{\mathbf{w}\in CD(\mathbf{u})}z_\mathbf{w} \right)^{-1} \right]\right].
		\end{eqnarray}
		Due to Assumption \ref{asn:nonmatchdist}, for every $1\leq c\leq n_B$, $(1,0,k,c)\in CD(\mathbf{u})$ if and only if $(1,0,k,1)\in CD(\mathbf{u})$. In other words, the following equality holds for every $1\leq c\leq n_B$:
		\begin{eqnarray}
			CD^-(1,0,k,1)&=&\left\{ \mathbf{u}\in Ind_{data}:(1,0,k,1)\in CD(\mathbf{u})\,\,\text{and}\,\, x_{\mathbf{u}}>0 \right\}\nonumber\\
			&=&\left\{ \mathbf{u}\in Ind_{data}:(1,0,k,c)\in CD(\mathbf{u})\,\,\text{and}\,\,x_{\mathbf{u}}>0 \right\}\nonumber\\
			&=&CD^-(1,0,k,c).
		\end{eqnarray}
		Therefore (\ref{eq:zmarginal_sum}) is equal to
		\begin{eqnarray}\label{eq:zmarginal_sum2}
			&&\sum_{c=1}^{n_B} \left[ \sum_{ \mathbf{u}\in CD^-(1,0,k,1) }  \left[ \left(x_\mathbf{u} z_{1,0,k,c} \right)  \left( \sum_{\mathbf{w}\in CD(\mathbf{u})}z_\mathbf{w} \right)^{-1} \right]\right]\nonumber\\
			&=& \sum_{ \mathbf{u}\in CD^-(1,0,k,1) }\left[ \sum_{c=1}^{n_B}  \left[ \left(x_\mathbf{u} z_{1,0,k,c} \right)  \left( \sum_{\mathbf{w}\in CD(\mathbf{u})}z_\mathbf{w} \right)^{-1} \right]\right]\nonumber\\
			&=& \sum_{ \substack{ \mathbf{u}\in Ind_{data} \\ \mathbf{u}:\, (1,0,k,1)\in CD(\mathbf{u}) \\\text{and }x_\mathbf{u}>0  } } x_\mathbf{u}\left( \sum_{\mathbf{w}\in CD(\mathbf{u})}z_\mathbf{w} \right)^{-1}\cdot \sum_{c=1}^{n_B}z_{1,0,k,c}\nonumber\\
			&=& \sum_{ \substack{ \mathbf{u}\in Ind_{data} \\ \mathbf{u}:\, k\in CD_{B=0}(\mathbf{u}) \\\text{and }x_\mathbf{u}>0  } } z_{1,0,k,:} x_\mathbf{u}\left( \sum_{\mathbf{w}\in CD(\mathbf{u})}z_\mathbf{w} \right)^{-1}.
		\end{eqnarray}
		Now suppose there is a $\mathbf{u}$ such that $(1,0,k,1)\in CD(\mathbf{u})$. Due to Assumption \ref{asn:distdivide}, the first two entries of $\mathbf{u}$ must be $(1,0)$, thus any element $\mathbf{w}\in CD(\mathbf{u})$ must be of the form $\mathbf{w}=(1,0,c,b)$ for some $(c,b)\in \{1,\dots,n_A\}\times \{1,\dots,n_B\}$, where by Assumption \ref{asn:nonmatchdist}, $(1,0,c,1)\in CD(\mathbf{u})$, meaning $c\in CD_{B=0}(\mathbf{u})$ by the definition of $CD_{B=0}$ as was defined in (\ref{eq:CD_marginal_def}). In other words,
		\begin{eqnarray}\label{eq:CD_marg_inclusion}
			CD(\mathbf{u})   \subseteq  \left\{ (1,0,c,b): c\in CD_{B=0}(\mathbf{u}),b\in\{1,\dots,n_B\} \right\}.
		\end{eqnarray}
		On the other hand, if $c\in CD_{B=0}(\mathbf{u})$, then $(1,0,c,1)\in CD(\mathbf{u})$ by definition, and $(1,0,c,b)\in CD(\mathbf{u})$ for every $b\in \{1,\dots,n_B\}$ by Assumption \ref{asn:distdivide}. This gives the opposite direction of the set inclusion in (\ref{eq:CD_marg_inclusion}) and making the two sets equal. This means the last expression in (\ref{eq:zmarginal_sum2}) is equal to 
		\begin{eqnarray}
			&&\sum_{ \substack{ \mathbf{u}\in Ind_{data} \\ \mathbf{u}:\, k\in CD_{B=0}(\mathbf{u}) \\\text{and }x_\mathbf{u}>0  } } z_{1,0,k,:} x_\mathbf{u} \left( \sum_{ \substack{ c\in CD_{B=0}(\mathbf{u}) \\ 1\leq d \leq n_B } } z_{1,0,c,d} \right)^{-1}\nonumber\\
			&=&\sum_{ \substack{ \mathbf{u}\in CD_{B=0}^-(k) }  } z_{1,0,k,:}x_\mathbf{u}\left( \sum_{c\in CD_{B=0}(\mathbf{u})}z_{1,0,c,:} \right)^{-1}
		\end{eqnarray}
		A very similar set of arguments will also lead to the equality
		\begin{eqnarray}
			z_{0,1,:,\ell}'=\sum_{ \mathbf{u}\in CD_{A=0}^-(\ell)  } z_{0,1,:,\ell}x_\mathbf{u}\left( \sum_{c\in CD_{A=0}(\mathbf{u})}z_{0,1,:,c} \right)^{-1}
		\end{eqnarray}
	\end{proof}

	\begin{lemma}\label{lem:DIST_marginal_fix}
		Suppose Assumptions \ref{asn:distdivide} and \ref{asn:nonmatchdist} are true, and that $\left( z_\mathbf{v} \right)_{\mathbf{v}\in Ind_{data}}$ satisfies \ref{asn:operational}. If $DIST_{B=0}$ fixes the marginal totals:
		\begin{equation}
			DIST_{B=0}\left( z_{1,0,1,:},\dots ,z_{1,0,n_A,:} \right)=\left( z_{1,0,1,:},\dots ,z_{1,0,n_A,:} \right),
		\end{equation}
		then $DIST$ fixes all values of $z_{1,0,k,\ell}$ for all $(k,\ell)\in \{1,\dots n_A\} \times \{1,\dots n_B\}$:
		\begin{equation}
			\pi_{1,0,k,\ell}\left( DIST\left( (z_\mathbf{v})_{\mathbf{v}\in Ind_{full}} \right) \right)=z_{1,0,k,\ell}.
		\end{equation}
		Similarly, if $DIST_{A=0}$ fixes the values of $z_{0,1,:,1},\dots,z_{0,1,:,n_B}$, then $DIST$ fixes $z_{0,1,k,\ell}$ for all $k$ and $\ell$.
	\end{lemma}
	\begin{proof}
		Suppose $1\leq k\leq n_A$ and $1\leq \ell\leq n_B$. Consider the case where $k\notin CD_{B=0}(\mathbf{u})$ for all $\mathbf{u}\in Ind_{data}$, then $(1,0,k,\ell)\notin CD(\mathbf{u})$ for all $\mathbf{u}\in Ind_{data}$, using the same logic used earlier in the proof of Lemma \ref{lem:DIST_marginals}. Then the $DIST$ function fixes the $(1,0,k,\ell)$ coordinate by default.
		\newline
		\newline
		Otherwise, suppose $k\in CD_{B=0}(\mathbb{u})$ for some $\mathbf{u}\in Ind_{data}$. As in the derivations of (\ref{eq:zmarginal_sum}) and (\ref{eq:zmarginal_sum2}):
		\begin{eqnarray}\label{eq:lem2_expansion}
			&&\pi_{1,0,k,\ell}\left( DIST\left( (z_\mathbf{v})_{\mathbf{v}\in Ind_{full}} \right) \right) \nonumber\\
			&=& z_{1,0,k,\ell}\sum_{\mathbf{u}\in CD^-(1,0,k,1)  }x_\mathbf{u} \left( \sum_{\mathbf{w}\in CD(\mathbf{u}) } z_\mathbf{w} \right)^{-1}\nonumber\\
			&=& z_{1,0,k,\ell} C_k
		\end{eqnarray}
		where
		\begin{equation}
			C_k=\sum_{\mathbf{u}\in CD^-(1,0,k,1)  }x_\mathbf{u} \left( \sum_{\mathbf{w}\in CD(\mathbf{u}) } z_\mathbf{w} \right)^{-1}.
		\end{equation}
		This would mean that the marginal totals of $DIST( (z_\mathbf{v})_{Ind_{full}} )$ would be 
		\begin{eqnarray}
			&&\pi_{1,0,k,1}\left( DIST\left( (z_\mathbf{v})_{\mathbf{v}\in Ind_{full}} \right) \right)+\dots+\pi_{1,0,k,n_B}\left( DIST\left( (z_\mathbf{v})_{\mathbf{v}\in Ind_{full}} \right) \right)\nonumber\\
			&=& C_k(z_{1,0,k,1}+\dots +z_{1,0,k,n_B})\nonumber\\
			&=& C_kz_{1,0,k,:}.
		\end{eqnarray}
		At the same time, the results of Lemma \ref{lem:DIST_marginals} shows that the marginal total equal
		\begin{eqnarray}
			&&\pi_{1,0,k,1}\left( DIST\left( (z_\mathbf{v})_{\mathbf{v}\in Ind_{full}} \right) \right)+\dots+\pi_{1,0,k,n_B}\left( DIST\left( (z_\mathbf{v})_{\mathbf{v}\in Ind_{full}} \right) \right)\nonumber\\
			&=& \pi_k  \left( DIST_{B=0}(z_{1,0,l,:},\dots z_{1,0,n_A,:}) \right)\nonumber\\
			&=& z_{1,0,k,:}
		\end{eqnarray}
		In order for $C_kz_{1,0,k,:}=z_{1,0,k,:}$ to be true, $C_k=1$ is necessary if $z_{1,0,k,:}$ is nonzero. In the case that $z_{1,0,k,:}=0$, then $z_{1,0,k,1}=\dots=z_{1,0,k,n_B}=0$ since all these values must be non-negative (per \ref{asn:operational}) and they sum to 0. In either case, (\ref{eq:lem2_expansion}) reduces to
		\begin{equation}
			\pi_{1,0,k,\ell}\left( DIST\left( (z_\mathbf{v})_{\mathbf{v}\in Ind_{full}} \right) \right)=z_{1,0,k,\ell}.
		\end{equation}
	\end{proof}

	\begin{lemma}\label{lem:poisson_solution}
		Suppose there are covariate vectors $\mathbf{v}_i\in \mathbb{R}^p$ and response values $y_i\in \mathbb{R}$ for $i=1,\dots,n$, positive integers $n\geq p$, with the span of $\{\mathbf{v}_i\}_{1\leq i\leq n}$ equaling $\mathbb{R}^p$. If there is a $\beta^*\in\mathbb{R}^p$ where
		\begin{equation}
			\exp\left( \mathbf{v}_i\cdot \beta^* \right)=y_i
		\end{equation}
		for $i=1,\dots,n$, then $\beta^*$ is the unique estimator when performing Poisson regression with $\mathbf{v}_i$'s and $y_i$'s.
	\end{lemma}
	\begin{proof}
		The estimator of the Poisson regression is the minimizer of the log-likelihood function, which is
		\begin{eqnarray}
			L\left(\beta\right)=\sum_{i=1}^n\bigg( y_i\beta\cdot \mathbf{v}_i-\exp(\beta\cdot \mathbf{v}_i)-\log(y_i!) \bigg).
		\end{eqnarray}
		The first and second derivatives are
		\begin{eqnarray}
			\nabla_\beta L&=&  \sum_{i=1}^n\bigg( y_i\mathbf{v}_i- \exp(\beta\cdot \mathbf{v}_i) \mathbf{v}_i \bigg)\nonumber\\
			\nabla^2_\beta L &=& -\sum_{i=1}^n \exp(\beta\cdot \mathbf{v}_i)\mathbf{v}_i\mathbf{v}_i^T 
		\end{eqnarray}
		If there is a $\beta^*$ such that $\exp\left( \mathbf{v}_i\cdot \beta^* \right)=y_i$ for all $i$, then the first derivative would be the 0 vector at $\beta^*$. The second derivative can be seen as negative definite; it is a linear combination of rank-1 matrices so $\mathbf{u}^T	\nabla^2_\beta L\,\mathbf{u}\leq 0$ for any nonzero $\mathbf{u}\in \mathbb{R}^p$, and $\mathbf{u}^T	\nabla^2_\beta L\,\mathbf{u}$ cannot equal $0$ since the $\mathbf{v}_i$'s linearly span all of $\mathbb{R}^p$. Therefore $L(\beta)$ has a unique maximizer at $\beta^*$.
	\end{proof}

	\printbibliography
\end{document}